\documentclass[12pt, leqno]{amsart}

\usepackage{graphicx}
\usepackage{amssymb,amsmath,amsthm,amscd}
\usepackage{mathrsfs}
\usepackage{lscape}
\usepackage{enumerate}
\usepackage[usenames,dvipsnames]{color}
\usepackage[colorlinks=true, pdfstartview=FitV, linkcolor=blue,citecolor=blue,urlcolor=blue]{hyperref}

\usepackage{tikz}
\usepackage[all]{xy}

\allowdisplaybreaks[3]

\numberwithin{equation}{section}

\newenvironment{red}{\relax\color{red}}{\relax}
\newenvironment{blue}{\relax\color{blue}}{\hspace*{.5ex}\relax}

\newcommand{\ber}{\begin{red}}
\newcommand{\er}{\end{red}}
\newcommand{\beb}{\begin{blue}}
\newcommand{\eb}{\end{blue}}

\theoremstyle{plain}
\newtheorem{lemma}{Lemma}[section]
\newtheorem{proposition}[lemma]{Proposition}
\newtheorem{theorem}[lemma]{Theorem}

\newtheorem{corollary}[lemma]{Corollary}

\theoremstyle{definition}
\newtheorem{remark}[lemma]{Remark}
\newtheorem{example}[lemma]{Example}

\newtheorem{definition}[lemma]{Definition}

\newcommand{\dl}{(\!(}
\newcommand{\dr}{)\!)}
\newcommand{\drl}{[\![}
\newcommand{\drr}{]\!]}
\newcommand{\vac}{\left| 0 \right>}
\newcommand{\n}{\mathfrak{n}}
\newcommand{\g}{\mathfrak{g}}

\newcommand{\C}{\mathbb{C}}

\newcommand{\Z}{\mathbb{Z}}
\newcommand{\N}{\mathbb{N}}

\newlength{\mylength}
\newcommand{\End}{{\rm{End}\ts}}
\newcommand{\Ker}{{\rm{Ker}}}

\newcommand{\ad}{{\rm{ad}}}

\newcommand{\non}{\nonumber}
\newcommand{\wh}{\widehat}

\newcommand{\La}{\Lambda}

\newcommand{\ts}{\,}

\newcommand{\tss}{\hspace{1pt}}

\newcommand{\CC}{\mathbb{C}\tss}

\newcommand{\Ac}{\mathcal{A}}

\newcommand{\Wc}{\mathcal{W}}

\newcommand{\gl}{\mathfrak{gl}}

\newcommand{\sll}{\mathfrak{sl}}

\newcommand{\cdet}{ {\rm cdet}}

\newcommand{\bari}{\bar{\imath}}
\newcommand{\barj}{\bar{\jmath}}

\newcommand{\fand}{\quad\text{and}\quad}
\newcommand{\Fand}{\qquad\text{and}\qquad}

\newcommand{\bth}{\begin{thm}}
\renewcommand{\eth}{\end{thm}}
\newcommand{\bpr}{\begin{prop}}
\newcommand{\epr}{\end{prop}}
\newcommand{\ble}{\begin{lm}}
\newcommand{\ele}{\end{lm}}
\newcommand{\bco}{\begin{cl}}
\newcommand{\eco}{\end{cl}}
\newcommand{\bex}{\begin{ex}}
\newcommand{\eex}{\end{ex}}
\newcommand{\bre}{\begin{rmk}}
\newcommand{\ere}{\end{rmk}}
\newcommand{\bcj}{\begin{conj}}
\newcommand{\ecj}{\end{conj}}

\newcommand{\bal}{\begin{aligned}}
\newcommand{\eal}{\end{aligned}}
\newcommand{\beq}{\begin{equation}}
\newcommand{\eeq}{\end{equation}}
\newcommand{\ben}{\begin{equation*}}
\newcommand{\een}{\end{equation*}}

\newcommand{\bpf}{\begin{proof}}
\newcommand{\epf}{\end{proof}}

\def\beql#1{\begin{equation}\label{#1}}

\title{Supersymmetric $W$-algebras}

\author[A.\,Molev]{Alexander Molev}
\address[A.\,Molev]
{School of Mathematics and Statistics,
University of Sydney,
NSW 2006, Australia}
\email{alexander.molev@sydney.edu.au}
\author
[E. Ragoucy]{Eric Ragoucy}
\address[E. Ragoucy]{Laboratoire de Physique Th\'{e}orique LAPTh,
CNRS, Universit\'{e} Savoie Mont Blanc and U.G.A.,
BP 110, 74941 Annecy-le-Vieux Cedex, France}
\email{eric.ragoucy@lapth.cnrs.fr}
\author[U.R. Suh]{Uhi Rinn Suh}
\address[U.R. Suh]
{ Department of Mathematical Sciences and Research institute of Mathematics, Seoul National University, GwanAkRo 1, Gwanak-Gu, Seoul 08826,
Korea}
\email{uhrisu1@snu.ac.kr}

\begin{document}

\begin{abstract}
We develop a general theory of $W$-algebras in the context of
supersymmetric vertex algebras. We describe the structure of $W$-algebras
associated with odd nilpotent elements of Lie superalgebras
in terms of their free generating sets. As an application, we produce
explicit free generators of the $W$-algebra associated with the odd principal
nilpotent element of the Lie superalgebra $\gl(n+1|n)$.

\medskip

\noindent
Preprint LAPTH-002/19

\end{abstract}

\maketitle

\section{Introduction}

The $W$-algebras first appeared in relation with the conformal field theory
in the work of Zamolodchikov~\cite{z:ie} and Fateev and Lukyanov~\cite{fl:mt}.
These algebras were studied intensively by physicists, both at the classical level
through Hamiltonian reduction of
 Wess--Zumino--Novikov--Witten models and their connection with affine Lie algebras, see e.g.
 \cite{Bow,Oraf,FRS}, but also using BRST formalism \cite{dBHTj,dBTj}.
For an extensive review on physicists works, see \cite{bs:ws}
and references therein. A definition of the $W$-algebras in the context of the vertex algebra theory and
quantized Drinfeld--Sokolov reduction was given by Feigin and Frenkel~\cite{ff:qd}; see also
the book by Frenkel and D. Ben-Zvi~\cite[Ch.~15]{fb:va}. A more general family of
$W$-algebras $W^k(\g,f)$ was introduced by
Kac, Roan and Wakimoto~\cite{krw:qr}, which depends on a simple Lie (super)algebra $\g$,
an (even) nilpotent element $f\in \g$ and the {\em level} $k\in\CC$. In the particular case of the
principal nilpotent element $f=f_{\text{\rm prin}}$ this reduces to the definition of \cite{ff:qd}; see also
a recent expository article by Arakawa~\cite{a:iw} where basic structure theorems and
representation theory of $W$-algebras are reviewed.

In the present paper we will be concerned with supersymmetric counterparts of the $W$-algebras
which can be defined by analogy with \cite[Ch.~15]{fb:va}.
Such $W$-algebras have already been studied, mostly in the physics literature; see
\cite{eh:st,i:qh,i:ns}. Moreover, a supersymmetric quantum hamiltonian reduction
approach was developed in the work of Madsen and the second author~\cite{MR94}.
We will rely on this work and the
supersymmetric vertex algebra theory developed by Heluani and Kac~\cite{HK06,K96}
to describe the structure of the $W$-algebras
associated with odd nilpotent elements of Lie superalgebras.
Our main structural result is Theorem~\ref{Thm_Sec_4} which describes
free generating sets of the $W$-algebras.

We will then apply the main result to
the case of the general linear Lie superalgebras. It is well-known that the
Lie superalgebra $\gl(m|n)$ contains an odd principal nilpotent element if and only if
$m=n\pm 1$. We take $m=n+1$ (this can be done without a real loss of generality)
and produce explicit free generators of the $W$-algebra as coefficients
of a certain noncommutative characteristic polynomial (Theorems~\ref{theorem1} and \ref{thm:freegen}).
These formulas can be regarded as supersymmetric analogues of the generators
of the principal $W$-algebra associated with the Lie algebra $\gl(n)$ produced by
Arakawa and the first author \cite{am:eg}. Furthermore, we show that the Miura transformation
used in \cite{am:eg} can also be applied in the supersymmetric context to recover
the generators of the $W$-algebra appeared in \cite{eh:st,i:qh,i:ns}.

\medskip

The second author wishes to thank the School of Mathematics and Statistics at the
University of Sydney for the hospitality and warm atmosphere during his visit,
as the work on this project was under way. The work of the third author
was supported by NRF Grant \# 2016R1C1B1010721.

\section{Supersymmetric Vertex Algebras}

In this section, we introduce supersymmetric vertex algebras following \cite{HK06} and \cite{K96}.
Proofs and additional details can be found in these references.
Note that in the terminology of the paper \cite{HK06} these objects are called $N_K=1$
{\em supersymmetric vertex algebras}.

\subsection{Notation and basic definitions}
We will be considering two couples of coordinates
\[ Z=(z, \theta), \quad W=(w, \zeta), \]
where $z$ and $w$ are even and $\theta$ and $\zeta$ are odd. Introduce the notation
\[ \C \drl Z \drr:= \C\drl z \drr \otimes \C[\theta], \quad
\C \dl Z \dr:= \C\dl z \dr \otimes \C[\theta]. \]
Since $\theta^2=0$ we have $\C[\theta]= \C\oplus \C\theta$. Similarly, 
\ben
\C[Z, Z^{-1}]:= \C[z,z^{-1}] \otimes \C[\theta], \quad \C[\![Z, Z^{-1}]\!]:=
\C[\![z,z^{-1}]\!] \otimes \C[\theta].
\een
Furthermore, set
\ben
\begin{aligned}
& Z-W := (z-w-\theta\zeta, \theta-\zeta), \\
& Z^{j_0|j_1} := z^{j_0} \theta^{j_1}  \quad  \text{ for } \ j_0\in \Z, \ j_1=0, 1, \\
& (Z-W)^{j_0|j_1} :=(z-w-\theta\zeta)^{j_0}(\theta-\zeta)^{j_1}.\\
\end{aligned}
\een

Let $\mathcal{U}=\mathcal{U}_{\bar{0}}\oplus \mathcal{U}_{\bar{1}}$ be
a $\Z/2\Z$-graded vector space which we will also call
a {\em vector superspace}. Accordingly, elements
$a\in \mathcal{U}_{\bar{0}}$ (resp. $a\in \mathcal{U}_{\bar{1}}$)
are called {\em even} (resp. {\em odd}) with the parity
$p(a)=\bar 0$ (resp. $p(a)=\bar 1$). The corresponding endomorphism algebra
$\End\mathcal{U}=(\End\mathcal{U})_{\bar{0}}\oplus
(\End\mathcal{U})_{\bar{1}}$ is a superalgebra, where
\ben
f\in (\End\mathcal{U})_{\bari}\quad \Longleftrightarrow \quad
f\big((\End\mathcal{U})_{\barj}\big)\subset (\End\mathcal{U})_{\bari+\barj}
\een
for any $\bari,\barj\in \Z/2\Z$.

Any element of the vector superspace
$\mathcal{U}[\![ Z, Z^{-1}]\!]:=
\mathcal{U} \otimes \C[\![ Z, Z^{-1}]\!]$ is called
a {\it $\mathcal{U}$-valued formal distribution}. It has the form
\begin{equation}\label{Formal dist}
a(Z)= \sum_{j_0\in \Z, \, j_1=0,1} Z^{j_0|j_1} a_{j_0|j_1}\in \mathcal{U}[\![ Z, Z^{-1}]\!],
\qquad a_{j_0|j_1}\in \mathcal{U}.
\end{equation}
The {\it super residue} of a formal distribution $a(Z)$ is defined by
\[ \text{res}_Z \,  a(Z):= a_{-1|1} \in \mathcal{U}.\]
Since
$ \text{res}_Z \, Z^{j_0|j_1} a(Z)= a_{-1-j_0|1-j_1}, $ it is convenient to use the notation
\[ a_{(j_0|j_1)}:= \text{res}_Z \, Z^{j_0|j_1} a(Z)\]
so that $a_{j_0|j_1}= a_{(-1-j_0|1-j_1)}$ and
the distribution $a(Z)$ in \eqref{Formal dist} takes the form
\[ a(Z)= \sum_{j_0\in \Z, \, j_1=0,1} Z^{-1-j_0|1-j_1} a_{(j_0|j_1)}.\]

An $\text{End}\, \mathcal{U}$-valued formal distribution $a(Z)$ is called
a {\it super field} if for any given $v\in \mathcal{U}$ there exists $N\in\Z_{\geqslant 0}$ such that
\[ a_{(j_0|j_1)}v=0 \quad\text{for all}\quad j_0 \geqslant N\,, \ j_1=0,1.\]

Similarly, a {\it $\mathcal{U}$-valued formal distribution in two variables} is an
element of the vector superspace $\mathcal{U} [\![ Z, Z^{-1}, W, W^{-1}]\!]$:
\[ a(Z,W)= \sum_{\substack{j_0, k_0 \in \Z,\\ \, j_1, k_1=0,1}}
Z^{j_0|j_1}W^{k_0|k_1}a_{j_0|j_1, k_0|k_1}\in \mathcal{U} [\![ Z, Z^{-1}, W, W^{-1}]\!]\]
with $ a_{j_0|j_1, k_0|k_1}\in \mathcal{U}$.  A formal
distribution $a(Z,W)$ is called {\it local} if
\[ (z-w)^n a(Z,W)=0\]
for some $n\in \Z_{\geqslant 0}.$
We let the {\it formal $\delta$-distribution} be defined by
\[ \delta(Z,W)= (\theta-\zeta) \sum_{n\in \Z} z^n w^{-n-1}.\]
Note that for any $f\in \mathcal{U}[\![ Z, Z^{-1}]\!]$ we have
\[ \text{res}_Z \delta(Z,W)f(Z) =f(W).\]
Since $(z-w) \delta(Z,W)=0$, the
formal $\delta$-distribution is local.

The differential operators $\partial_z$, $\partial_\theta$, $\partial_w$
and $\partial_\zeta$ act naturally on $\C[\![Z,Z^{-1}, W, W^{-1} ]\!]$.
Consider two more odd differential operators
\ben
D_Z= \partial_\theta+\theta \partial_z, \quad D_W=
\partial_\zeta +\zeta \partial_w.
\een
Then $[D_Z, D_Z]=2\partial_z$. Set
\[ D_Z^{j_0|j_1}= \partial_z^{j_0} D_Z^{j_1}, \quad D_Z^{(j_0|j_1)}=(-1)^{j_1}
\frac{1}{j_0 !} D_Z^{j_0|j_1}.\]

\begin{lemma}\label{Decomposition lemma}
Let $a(Z,W)$ be a local formal distribution. Then
\[ a(Z,W)= \sum_{\substack{j_0\in \Z_{\geqslant 0},\\  j_1=0,1}}
D_W^{(j_0|j_1)} \delta(Z,W)\ts c_{j_0|j_1}(W),
\]
where the sum is finite, and
\[ c_{j_0|j_1}(W)= \text{res}_Z (Z-W)^{j_0|j_1} a(Z,W).\]
\end{lemma}

\begin{definition}\label{SUSYVA}
A {\em supersymmetric vertex algebra} is a tuple $(V, \vac, S, Y)$
where $V$ is a vector superspace, $\vac\in V$ is a {\em vacuum vector},
$S$ is an odd endomorphism of $V$, and the {\em state-field correspondence} $Y$
is a parity preserving linear map from $V$ to the space of
$\End V$-valued super fields
\ben
 Y: V \to \End V[\![ Z,Z^{-1} ]\!] , \quad a\mapsto a(Z)
 \een
satisfying the following axioms:
\begin{itemize}
\item ({\em vacuum}) $a(Z)\vac |_{z=0,\ts \theta=0}=a,\ts S\vac=0$,
\item ({\em translation covariance}) $[S, a(Z)]=(\partial_\theta-\theta \partial_z) a(Z)$,
\item ({\em locality}) for any $a,b\in V$ there exists $N \in \Z_+$ such that\\
$(z-w)^N [a(Z), b(W)]=0$.
\end{itemize}
\end{definition}

By Lemma~\ref{Decomposition lemma}, the locality axiom implies
a finite sum decomposition
\ben
\, [a(Z), b(W)]=\sum_{\substack{j_0\in \Z_{\geqslant 0},\\ \ts j_1=0,1}}
\big(D_W^{(j_0|j_1)} \delta(Z,W)\big) \, a(W)_{(j_0|j_1)}b(W)  \,
\een
for  $a(W)_{(j_0|j_1)}b(W):= \text{res}_Z (Z-W)^{j_0|j_1} [a(Z),b(W)]$.
The expression $a(W)_{(j_0|j_1)}b(W)$ is called the {\em $(j_0|j_1)$-th product of
the super fields $a(W)$ and $b(W)$}.

\begin{definition}\label{Def:products}
\begin{enumerate}
\item The {\it normally ordered product} of two $\End V$-valued
formal distributions $a(Z)$ and $b(Z)$ is defined by
\[ :a(Z)b(Z):= a_+(Z)b(Z)+(-1)^{p(a)p(b)} b(Z) a_-(Z),\]
where
\ben
a_+(Z)= \sum_{j_0\in \Z_{\geqslant 0},\ts j_1=0,1} Z^{j_0|j_1} a_{j_0|j_1}\fand
a_-(Z)= \sum_{j_0\in \Z_{< 0},\ts j_1=0,1} Z^{j_0|j_1} a_{j_0|j_1}.
\een
\item If $j_0\leqslant -2$ and $j_1=0,1$, or $j_0=-1$ and $j_1=0$, then $a(Z)_{(j_0|j_1)} b(Z)$
is  given by
\[ a(Z)_{(j_0|j_1)} b(Z)= (-1)^{1-j_1} :
\big( D_Z^{(-1-j_0|1-j_1)} a(Z)\big) b(Z): .\]
\end{enumerate}
\end{definition}

\begin{remark} \label{Rem:products}
One can check that
\[ :a(Z)b(Z):\vac|_{z=0, \theta=0}= a_{(-1|1)}b\]
and
\[ a(Z)_{(j_0|j_1)}b(Z) \vac |_{z=0, \theta=0} = a_{(j_0|j_1)}b\]
for $(j_0,j_1)$ as in part (2) of Definition \ref{Def:products}.
\end{remark}

\begin{lemma}[Dong's lemma]
Let $a(Z), b(Z), c(Z)$ be pairwise local formal distributions.
Then $\big(a(Z), (b_{(j_0|j_1)}c)(Z)\big)$ is local for any $j_0\in \Z$ and $j_1=0,1$.
\end{lemma}

\begin{lemma}[Uniqueness lemma]
Let $V$ be a supersymmetric vertex algebra. If $a(Z)$ is a super field
such that $(a(Z), b(Z))$ is local for every $b\in V$
and $a(Z) \vac=0$ then $a(Z)=0$.
\end{lemma}

By the uniqueness lemma and Remark \ref{Rem:products},
\[ a(Z)_{(j_0|j_1)} b(Z)= (a_{(j_0|j_1)} b)(Z),\]
and  we set
\[ :ab: \, = a_{(-1|1)}b=\,  :a(Z)b(Z): \vac|_{z=0,\ts \theta=0}.\]

Note that for a given supersymmetric vertex algebra $V$, the state-field correspondence map
\[ Y: V \to (\text{End} \, V) [\![Z,Z^{-1}]\!], \quad a \mapsto a(Z), \]
is injective. Hence a supersymmetric vertex algebra $V$ can be considered
as a set of super fields $Y(V)$. In the following theorem, we construct
a vertex algebra as a set of super fields.

\begin{theorem}[Existence theorem]
Let $V$ be a vector superspace and $\wh V$ be a set of pairwise
local $\End V$-valued super fields. Suppose $\text{Id}\in \wh V$
is the constant field and $\wh V$ is invariant under the operator
$D=\partial_\theta+ \theta \partial_z$  and all $(j_0|j_1)$-products.
Then the superspace $V$ with the vacuum vector $\text{Id}$, the operator $S$ given by
$Sa(Z)=  D(a(Z))$ and the $(j_0|j_1)$-products
is a supersymmetric vertex algebra.
\end{theorem}

\subsection{Supersymmetric Lie conformal algebras} \label{sec:LCA}

Recall that a Lie conformal algebra (LCA) $R$ gives rise to a vertex
algebra called  a universal enveloping vertex algebra $V(R)$ \cite{BK03,K96}.
Now we introduce its supersymmetric analogue: that is, a
supersymmetric LCA and the corresponding
universal enveloping supersymmetric vertex algebra.
Consider two superalgebras:
\begin{itemize}
\item Let $\mathcal{L}$ be the associative superalgebra
generated by a pair of elements $\Lambda=( \lambda, \chi )$, where $\lambda$
is even and $\chi$ is odd, such that
\[ \, [\lambda, \chi]=0, \quad [\chi, \chi]=2\chi^2=-2\lambda. \, \]
\item Let $\mathcal{K}$ be another associative superalgebra
generated by a pair of elements $\nabla=(T,S)$, where $T$
is even and $S$ is odd, such that
\[ \, [T,S]=0, \quad [S, S]=2S^2= 2T.\]
\end{itemize}
Note that $\mathcal{L}$ and $\mathcal{K}$ are isomorphic via
the map $\lambda\mapsto -T$ and $\chi\mapsto -S$.

Set
\ben
 (Z-W) \Lambda = (z-w-\theta\zeta)\lambda + (\theta-\zeta)\chi.
 \een
Given a formal distribution $a(Z,W)$ of two variables $Z$ and $W$,
consider the {\em formal Fourier transformation}
\ben
\mathcal{F}^\Lambda_{Z,W} \, a(Z,W)= \text{res}^{}_Z \, \text{exp}\big((Z-W)\Lambda\big) a(Z,W)
\een
which can be expanded as
\ben
\mathcal{F}^\Lambda_{Z,W} \, a(Z,W) = \sum_{\substack{j_0\in \Z_{\geqslant 0}, \ts
j_1=0,1}} (-1)^{j_1} \Lambda^{(j_0|j_1)} c_{j_0|j_1}(W),
\een
where
\ben
\Lambda^{(j_0|j_1)} = (-1)^{j_1} \frac{\lambda^{j_0} \chi^{j_1}}{j_0 !}
\een
and $c_{j_0|j_1}(W)$ is defined in Lemma \ref{Decomposition lemma}.

Define the $\Lambda$-{\em bracket} $(a,b)\to[a_\Lambda b]$ of a
local pair $\big(a(Z), b(Z)\big)$ by
\ben
 [a_\Lambda b](W):= \mathcal{F}^\Lambda_{Z,W}[a(Z), b(W)].
 \een

\begin{proposition}\label{prop:Lambda-bracket}
The $\Lambda$-bracket satisfies the following
properties for all pairwise local distributions $(a(Z), b(Z), c(Z))$:
\begin{enumerate}
\item (sesquilinearity) \[ \, [Sa_\Lambda b]= \chi[a_\Lambda b],\quad
[a_\Lambda Sb]= -(-1)^{p(a)}(S+\chi)[a_\Lambda b];\, \]
\item (skew-symmetry) \[ \, [b_\Lambda a] = (-1)^{p(a)p(b)}
[a_{-\Lambda-\nabla} b],\, \]
where
\ben
[a_{-\Lambda-\nabla} b]= \sum_{j_0\in \Z_{\geqslant 0},\ts j_1=0,1}
(-1)^{j_1}(-\Lambda-\nabla)^{(j_0|j_1)} a_{(j_0|j_1)}b
\een
for
$-\Lambda-\nabla= (-\lambda-T, -\chi-S)$ with
\ben
[\chi, S]= 2\lambda\Fand [\chi, T]=[\lambda, T]=[\lambda, S]=0;
\een
\item (Jacobi identity)
\[ \, [ a_\Lambda [ b_\Gamma c]]= -(-1)^{p(a)}[[a_\Lambda b]_{\Lambda+\Gamma} c]
+(-1)^{(p(a)+1)(p(b)+1)} [b_\Gamma[a_\Lambda c]],\]
where
\begin{enumerate}[(i)]
\item $\Gamma=(\gamma, \eta)$ with $[\gamma, \eta]=[\gamma, \gamma]=0$
and $[\eta,\eta]=-2\gamma$,
\item$\Lambda+\Gamma=(\lambda+\gamma, \zeta+\eta)$ with
$[\lambda, \eta]=[\lambda,\gamma]=[\zeta,\gamma]= [\zeta, \eta]=0$.
\end{enumerate}
\end{enumerate}
\end{proposition}

This motivates the following definition.

\begin{definition}\label{Def:susyLCA}
A {\it supersymmetric Lie conformal algebra (LCA)} $\mathcal{R}$
is a $\Z/2\Z$-graded $\mathcal{K}$-module  endowed with
odd bilinear map $\mathcal{R} \otimes \mathcal{R} \to \mathcal{L} \otimes \mathcal{R}$,
called $\Lambda$-bracket, given by a finite sum expansion
\[ a\otimes b \mapsto [a_\Lambda b] = \sum_{\substack{j_0\in \Z_{\geqslant 0},\ts j_1=0,1 }}
(-1)^{j_1} \Lambda^{(j_0|j_1)} a_{(j_0|j_1)} b\]
with $a_{(j_0|j_1)}b \in \mathcal{R}$,
satisfying the following properties:
\begin{enumerate}
\item (sesquilinearity) In $\mathcal{L}\otimes \mathcal{R}$ we have
\[ \, [Sa_\Lambda b]= \chi[a_\Lambda b],\quad
[a_\Lambda Sb]= -(-1)^{p(a)}(S+\chi)[a_\Lambda b]\, , \]
where $S$ and $\chi$ obey the relation $[S, \chi]= 2\lambda$;
\item (skew-symmetry) In $\mathcal{L}\otimes \mathcal{R}$ we have
\[ \, [b_\Lambda a] = (-1)^{p(a)p(b)} [a_{-\Lambda-\nabla} b],\, \]
where
\ben
[a_{-\Lambda-\nabla} b]= \sum_{j_0\in \Z_{\geqslant 0},\ts j_1=0,1}
(-1)^{j_1}(-\Lambda-\nabla)^{(j_0|j_1)} a_{(j_0|j_1)}b
\een
for
$-\Lambda-\nabla= (-\lambda-T, -\chi-S)$ satisfying
\ben
[\chi, S]= 2\lambda\Fand [\chi, T]=[\lambda, T]=[\lambda, S]=0;
\een
\item (Jacobi-identity) In $\mathcal{L}\otimes \mathcal{L}' \otimes \mathcal{R}$ we have
\[ \, [ a_\Lambda [ b_\Gamma c]]= -(-1)^{p(a)}[[a_\Lambda b]_{\Lambda+\Gamma} c]
+(-1)^{(p(a)+1)(p(b)+1)} [b_\Gamma[a_\Lambda c]],\]
where
\begin{enumerate}[(i)]
\item $\Gamma=(\gamma, \eta)$ such that $[\gamma, \eta]=[\gamma, \gamma]=0$
and $[\eta,\eta]=-2\gamma$,
\item$\Lambda+\Gamma=(\lambda+\gamma, \zeta+\eta)$ such that
$[\lambda, \eta]=[\lambda,\gamma]=[\zeta,\gamma]= [\zeta, \eta]=0$.
\end{enumerate}
\end{enumerate}
\end{definition}

Note that the tensor product sign is often omitted in the notation.

The next theorem provides an equivalent definition of supersymmetric vertex algebras in terms
of $\Lambda$-brackets; cf. \cite[Thm.~4.1]{k:iv}.

\begin{theorem}A supersymmetric vertex algebra is a tuple
$(V, S,  [\, _\Lambda\, ], \vac, : \, \, :)$ such that
\begin{enumerate}[(i)]
\item $(V, S,  [\, _\Lambda\, ])$ is a supersymmetric Lie conformal algebra.
\item $(V, S, \vac, : \, \, :)$ is a unital
differential superalgebra, where $S$ is an odd derivation of the product $: \, \, :$,
and the following properties hold:
\begin{equation}\label{Eqn:q-comm/assoc}
\begin{aligned}
&  :ab:-(-1)^{p(a)p(b)} :ba:= (-1)^{p(a)p(b)}
\sum_{j\geqslant 1} \frac{(-T)^j}{j!} (b_{(-1+j|1)}a),  \\
&  ::ab:c:-:a:bc::= \sum_{j\geqslant 0} a_{(-2-j|1)}(b_{(j|1)}c)+(-1)^{p(a)p(b)}
\sum_{j\geqslant 0} b_{(-2-j|1)}(a_{(j|1)}c).
\end{aligned}
\end{equation}
\item The $\Lambda$-bracket and the product $: \, \, :$ are related by the
non-commutative Wick formula\ts:
\begin{equation} \label{Wick}
 \, [a_\Lambda :bc:] =\sum_{k\geqslant 0} \frac{\lambda^k}{k!} [a_\Lambda b]_{(k-1|1)}c
 + (-1)^{(p(a)+1)p(b)} :b[a_\Lambda c]: \, .
 \end{equation}
\end{enumerate}
\end{theorem}

The properties \eqref{Eqn:q-comm/assoc} of the product $: \, \, :$ are referred to as the
{\em quasi-commutativity} and {\em quasi-associativity}, respectively.

\begin{definition}
\begin{enumerate}
\item A set $\mathcal{B}=\{a_i\ |\ i\in I\}$  of elements in  a supersymmetric
vertex algebra $V$ {\it strongly generates} $V$ if the set of  monomials
\[\{\, :a_{j_1} a_{j_2} \dots a_{j_s}: \,  |\,  j_1, \dots, j_s\in I, \, s\in \Z_{\geqslant 0}\}\]
spans $V$. If $s=0$, the monomial is understood as $\vac$. For $s>2$
the product in the monomial is applied consecutively from right to left.
\item An ordered set $\mathcal{B} =\{ a_i\ |\ i\in I\}\subset V $
{\it freely
generates}  a supersymmetric vertex algebra $V$ if the set of monomials
\[\{\,  :a_{j_1} a_{j_2}\dots a_{j_s}:\, | \,  j_r \leqslant j_{r+1}
\text{ and } j_r < j_{r+1} \text{ if  }
p(a_{j_r}) = \bar 1 \}\]
 forms a basis of $V$ over $\C$.
\end{enumerate}
\end{definition}

\begin{theorem}\label{Thm:universal SUSY VA}
Let $\mathcal{R}$ be a supersymmetric Lie conformal algebra with
an ordered $\C$-basis $\mathcal{B}=\{ a_i\ |\ i\in I\}.$ Then there
exists a unique supersymmetric vertex algebra $V(\mathcal{R})$ such that
\begin{enumerate}[(i)]
\item $V(\mathcal{R})$ is freely generated by $\mathcal{B}$,
\item the operator $S$ on $V(\mathcal{R})$ is defined by $S(:ab:)=:(Sa)b: +(-1)^{p(a)} :a(Sb):$,
\item the $\Lambda$-bracket on $\mathcal{R}$ extends
to the $\Lambda$-bracket on $V(\mathcal{R})$ via the Wick formula \eqref{Wick}.
\end{enumerate}
\end{theorem}

\begin{definition}
\label{def:uea}
For a given supersymmetric Lie conformal algebra $\mathcal{R}$,
the supersymmetric vertex algebra $V(\mathcal{R})$ in
Theorem \ref{Thm:universal SUSY VA} is called
the {\it universal enveloping supersymmetric vertex algebra}
associated to $\mathcal{R}$.
\end{definition}

\subsection{Supersymmetric nonlinear LCAs}
\label{subsec:snlca}
In this section we follow Section 3 of \cite{DK06} to introduce {\em nonlinear} supersymmetric LCAs.
We omit the arguments which
are straightforward supersymmetric analogues of those in \cite{DK06}.

 For a positive integer $n$, consider a $\mathcal{K}$-module
 $\mathcal{R}= \bigoplus_{\zeta\in \N/n} \mathcal{R}_\zeta$ with $(\N/n)$-grading
 so that $\text{gr}(a)=\zeta$ for $a\in \mathcal{R}_\zeta$.
 The grading $\text{gr}$ is naturally extended to the grading of
 the tensor algebra $\mathcal{T}(\mathcal{R})$ by
\[ \text{gr}(a\otimes b)= \text{gr}(a)+ \text{gr}(b).\]
Set
\[ \mathcal{T}(\mathcal{R}) _{(\zeta)-}
= \bigoplus_{\zeta'<\zeta}\mathcal{T}(\mathcal{R})_{\zeta'}.\]

\begin{definition}
Suppose that
$\mathcal{R}$ is endowed with a {\it nonlinear $\Lambda$-bracket}
\ben
 [\mathcal{R}_{\zeta\, \Lambda\, } \mathcal{R}_{\zeta'}] \subset
 \mathcal{L}\otimes \mathcal{T}(\mathcal{R})_{(\zeta+\zeta')_-},
 \een
 satisfying skew-symmetry, sesquilinearity and Jacobi identity in
 Definition \ref{Def:susyLCA}. Then $\mathcal{R}$ is
 called {\it supersymmetric nonlinear
 Lie conformal algebra.}
\end{definition}

\begin{proposition} Let $\mathcal{R}$ be a supersymmetric nonlinear LCA.
Then the normally ordered product and $\La$-bracket admit unique extensions
to the linear maps
\ben
\begin{aligned}
& \mathcal{T}(\mathcal{R}) \otimes \mathcal{T}(\mathcal{R}) \to
\mathcal{T}(\mathcal{R}), \quad A\otimes B \mapsto :AB:, \\
&  \mathcal{T}(\mathcal{R}) \otimes \mathcal{T}(\mathcal{R}) \to
\mathcal{L} \otimes \mathcal{T}(\mathcal{R}), \quad A\otimes B \mapsto [A_\Lambda B],
\end{aligned}
\een
in such a way that for any $a,b\in \mathcal{R}$ and $A,B,C\in \mathcal{T}(\mathcal{R})$ we have
\begin{enumerate}[(i)]
\item $[a_\Lambda b]$ is defined by the $\Lambda$-bracket on $\mathcal{R}$,
\item $:aB:= a\otimes B$,
\item $:1A:=:A1:=A$,
\item $:(a\otimes B) C:- :a:BC::$ is defined by the quasi-associativity,
\item $[A_\Lambda(b\otimes C)]$ and $[(a\otimes B)_\Lambda C]$
are defined by the Wick formula.
\end{enumerate}
\end{proposition}

For a given supersymmetric nonlinear LCA $\mathcal{R}$, consider
the two-sided ideal $\mathcal{J}(\mathcal{R})$ of $\mathcal{T}(\mathcal{R})$ generated by
elements of the form
\[ (:ab:-(-1)^{p(a)p(b)}:ba:)- (-1)^{p(a)p(b)} \sum_{j\geqslant 1} \frac{(-T)^j}{j!} b_{(-1+j|1) a},\]
where
\ben
[b_\Lambda a]= \sum_{j_0\in \Z_{\geqslant 0},\ts j_1=0,1}(-1)^{j_1}\Lambda^{(j_0|j_1)} b_{(j_0|j_1)}a.
\een
Then the $\La$-bracket and the product $: \ \ :$ on
$\mathcal{T}(\mathcal{R})$ induce a well-defined $\La$-bracket and product on the quotient
\[ V(\mathcal{R})= \mathcal{T}(\mathcal{R})/\mathcal{J}(\mathcal{R}).\]
Since $V(\mathcal{R})$ satisfies quasi-commutativity, quasi-associativity and Wick formula,
it is a supersymmetric vertex algebra which is
called the {\it universal enveloping supersymmetric vertex algebra of $\mathcal{R}$};
cf. Definition~\ref{def:uea}.

\begin{proposition}
For a given ordered basis $\mathcal{B}$ of $\mathcal{R}$,
the supersymmetric vertex algebra $V(\mathcal{R})$
is freely generated by $\mathcal{B}$.
\end{proposition}

\section{Good filtered complexes of supersymmetric nonlinear LCAs}
\label{Sec:filtered complex}

Here we reproduce some useful facts about bigraded complexes.
Proofs can be obtained by suitable supersymmetric versions
of the arguments in \cite[Sec.~4]{DK06}. Introduce the notation
\[ \Gamma= \frac{\Z}{2}, \quad \Gamma_+= \frac{\Z_{\geqslant 0}}{2}, \quad
\Gamma'_+= \frac{\Z_{>0}}{2}. \]
Let $\g$ be a graded vector superspace and $\mathcal{R}= \mathcal{K} \otimes \g$
be a nonlinear Lie conformal algebra such that
\begin{equation}\label{Eqn:bigrading}
 \g= \bigoplus_{\substack{p,q\in \Gamma,\  p+q=\Z_+, \\ \Delta\in \Gamma'_+}}
 \g^{p,q}[\Delta], \qquad \mathcal{R}= \bigoplus_{\substack{p,q\in \Gamma,\
 p+q=\Z_+, \\ \Delta\in \Gamma'_+}}  \mathcal{R}^{p,q}[\Delta],
 \end{equation}
where
\ben
\mathcal{R}^{p,q}[\Delta]= \bigoplus_{n\geqslant 0}
S^n\otimes \g^{p,q}\big[\Delta-\frac{n}{2}\big].
\een.

The universal enveloping supersymmetric vertex algebra $V(\mathcal{R})$,
which is strongly generated by a basis  $\{a_i\ts |\ts i\in I\}$  of $\mathcal{R}$,
has the $\Gamma'_+$-grading
\[ V(\mathcal{R})= \bigoplus_{\Delta\in \Gamma'_+}V(\mathcal{R})[\Delta]\]
where
\[  V(\mathcal{R})[\Delta]= \textstyle
\text{span}_\C\{\, :a_{i_1}a_{i_2}\dots a_{i_s}: \, |   \, i_k\in I,\, a_{i_k}
\in \mathcal{R}[\Delta_k],  \,  \sum_{k=1}^s \Delta_k=\Delta\}.\]
We assume that
\[ V(\mathcal{R})[\Delta_1] _{(n_0|n_1)} V(\mathcal{R})[\Delta_1]
\subset V(\mathcal{R})\big[\Delta_1+\Delta_2 -n_0-\frac{n_1}{2} -\frac{1}{2}\big].\]

Consider a $\Gamma$-filtration and a $\Z$-grading of $\mathcal{R}$
induced from \eqref{Eqn:bigrading}
\[ F^p \mathcal{R}= \bigoplus_{\substack{p'\geqslant p, \\  q, \Delta}}
\mathcal{R}^{p', q}[\Delta], \qquad \mathcal{R}^n = \bigoplus_{p+q=n} \mathcal{R}^{p,q},\]
and the corresponding filtration and $\Z$-grading of $V(\mathcal{R})$ defined by
\ben
\begin{aligned}
& V(\mathcal{R})^n = \textstyle \text{span}_\C\{\, :a_{i_1}a_{i_2}
\dots a_{i_s}: \, |   \, i_k\in I,\, a_{i_k} \in \mathcal{R}^{p_k, q_k},  \,
\sum_{k=1}^s p_k+q_k=n\},\\[0.5em]
&    F^p V(\mathcal{R})= \textstyle \text{span}_\C\{\, :a_{i_1}a_{i_2}\dots a_{i_s}: \, |
\, i_k\in I,\, a_{i_k} \in \mathcal{R}^{p_k, q_k},  \,  \sum_{k=1}^s p_k\geqslant p\}.
\end{aligned}
 \een
Set
 \[ F^pV(\mathcal{R})^n = F^pV(\mathcal{R}) \cap V(\mathcal{R})^n,
 \quad F^pV(\mathcal{R})^n[\Delta]= F^pV(\mathcal{R})^n \cap V(\mathcal{R})[\Delta] \]
and consider the associated graded algebra
 \[ \text{gr}\ts V(\mathcal{R}) = \bigoplus_{p,q\in \Gamma} \text{gr}^{p,q}V(\mathcal{R}),\]
where
 \ben
 \begin{aligned}
 & \text{gr}^{p,q} V(\mathcal{R})[\Delta]= F^p V(\mathcal{R})^{p+q}
 [\Delta]/F^{p+\frac{1}{2}} V(\mathcal{R})^{p+q}[\Delta],\\[0.4em]
 & \text{gr}^{p,q} V(\mathcal{R}) = F^p V(\mathcal{R})^{p+q}/ F^{p+\frac{1}{2}}
 V(\mathcal{R})^{p+q}= \bigoplus_{\Delta\in \Gamma'_+}  \text{gr}^{p,q} V(\mathcal{R})[\Delta].
 \end{aligned}
\een
\vskip 2mm
Suppose a differential map $ d: V(\mathcal{R}) \to V(\mathcal{R})$ satisfies
\begin{equation}\label{d:level1}
 d( F^pV(\mathcal{R})^n) \subset F^p V(\mathcal{R})^{n+1}, \quad
 d(V(\mathcal{R}[\Delta]) \subset V(\mathcal{R})[\Delta].
 \end{equation}
 Then we set for the cohomology spaces
 \ben
\begin{aligned}
&  F^pH^n(V(\mathcal{R}), d) = \Ker(d|_{F^pV(\mathcal{R})^n}) /\text{Im}\ts
d \cap F^p V(\mathcal{R})^n, \\[0.4em]
& \text{gr}^{p,q}H(V(\mathcal{R}), d)= F^p H^{p+q}(V(\mathcal{R}), d) / F^{p+\frac{1}{2}}
H^{p+q}(V(\mathcal{R}), d).
\end{aligned}
\een
In addition, for  the graded differential map $d^{\tss\text{gr}}: \text{gr}\ts V(\mathcal{R})
\to\text{gr}\ts V(\mathcal{R})$ induced from $d$, we define cohomology spaces by
\ben
 H^{p,q}(\text{gr}\ts V(\mathcal{R}), d^{\tss\text{gr}}) =
 \Ker\ts d^{\tss\text{gr}}|_{\text{gr}^{p,q}V(\mathcal{R})} /\text{Im}\ts d^{\tss\text{gr}}
 \cap \text{gr}^{p,q}V(\mathcal{R}).
\een

\begin{definition}
Let $d$ be a differential on $V(\mathcal{R})$ satisfying \eqref{d:level1}.
\begin{enumerate}
\item  We say $d$ is {\it almost linear differential of $\mathcal{R}$} if
\[ d^{\tss\text{gr}}( \g^{p,q}[\Delta] ) \subset \g^{p,q+1}[\Delta];\]
or, equivalently, $ d(\g^{p,q}[\Delta]) \subset \g^{p,q+1}[\Delta]
\oplus F^{p+\frac{1}{2}} V(\mathcal{R})^{p+q+1}$.
\item A differential $d$ is called a {\it good} almost linear differential of $\mathcal{R}$ if
\[ H^{p,q}(\g, d^{\tss\text{gr}})=0 \quad \text{ if } \quad p+q\neq 0.\]
\end{enumerate}
\end{definition}

In the rest of this section we assume that  $ V(\mathcal{R})[\Delta]$
has  finite dimension  for any $\Delta \in \Gamma'_+$ and $d$ is a good
almost linear differential of $\mathcal{R}$.  Take bases
\ben
\begin{aligned}
& \mathcal{B}_{\g}^p[\Delta]= \{\, e_i\, |\, i\in \mathcal{I}_\g^p[\Delta]\, \}
\qquad\text{for some index sets}\ \  \mathcal{I}_\g^p[\Delta],\\[0.2em]
&  \mathcal{B}_{\mathcal{R}}^p[\Delta] = \{ e_{(i,n)} \, | \, e_{(i,n)}= S^n e_i, \, e_i
\in \mathcal{B}_{\g}^{p}[\Delta'], \,  \Delta' +\frac{n}{2} =\Delta\},
\end{aligned}
\een
of $\g^{p,-p}[\Delta] \cap \Ker\ts d^{\tss\text{gr}} $  and $\mathcal{R}^{p,-p}[\Delta]
\cap \Ker\ts d^{\tss\text{gr}}= H^{p,-p}(\text{gr}\ts\mathcal{R}, d^{\tss\text{gr}})[\Delta]$,
respectively.
Then
\begin{equation*}
\mathcal{B}_{\mathcal{R}}:=\bigsqcup_{\substack{\Delta\in \Gamma'_+,\ts  p\in \Gamma }}
\mathcal{B}_{\mathcal{R}}^p[\Delta]= \{\,  e_{(i,n)}\, | \,
e_{(i,n)}= S^n e_i, \, i\in \mathcal{I}_\g \, \}
\end{equation*}
is a basis of
$ H(\text{gr}\ts\mathcal{R}, d^{\tss\text{gr}})$,
where
\ben
\mathcal{I}_\g:= \bigsqcup_{\substack{\Delta\in \Gamma'_+,\ts  p\in \Gamma }}
\mathcal{I}_\g^p[\Delta].
\een

\begin{proposition}\hfill
\begin{enumerate}
\item $H(\text{gr}\ts V(\mathcal{R}), d^{\tss\text{gr}})$ is freely
generated by $\mathcal{B}_{\mathcal{R}}$.
\item $H^{p,-p}(\text{gr}\ts V(\mathcal{R}), d^{\tss\text{gr}})[\Delta]$ has
the basis
\[  \mathcal{B}_{V(\mathcal{R})}^p[\Delta]= \big\{ : e_{(i_1, n_1)}e_{(i_2, n_2)}
\dots e_{(i_k, n_k)}:\big\},
\]
where the sets of indices $( i_t , n_t) \in
\mathcal{I}_\g^{p_t}[\Delta_t]\times \Z_{\geqslant 0} $
satisfy the conditions:
\begin{enumerate}[(i)]
\item $(i_t, n_t)\leqslant (i_{t+1}, n_{t+1})$,
\item if
$e_{(i_t, n_t)}$ and $e_{(i_{t+1}, n_{t+1})}$  are odd then $(i_t, n_t)<(i_{t+1}, n_{t+1})$,
\item $\sum_{t=1}^{k} i_t= p,$
\item $\sum_{t=1}^k \big(\Delta_t+ \frac{n_t}{2}\big)= \Delta
.$
\end{enumerate}
\end{enumerate}
\end{proposition}

For $e_i \in \g^{p,-p}[\Delta] \cap \Ker\ts d^{\tss\text{gr}}$ there
exists an element $f_i\in F^{p+\frac{1}{2}}V(\mathcal{R})^0[\Delta]$ such
that $E_i=e_i+ f_i \in F^p V(\mathcal{R})^0[\Delta] \cap \Ker\ts d$.
Set
\[ H^{p,-p}(\g, d)[\Delta]= \text{span}\ts\{ \, E_i\, |\, i\in
\mathcal{I}_{\g}^p[\Delta]\, \}, \quad H(\g, d)[\Delta]= \bigoplus_{p\in \Gamma}
H^{p,-p}(\g, d)[\Delta] .\]

\begin{theorem} \label{Thm:cohomology}\hfill
\begin{enumerate}
\item $H(V(\mathcal{R}), d) = H^0 (V(\mathcal{R}), d).$
\item If the $\mathcal{K}$-module $H(\mathcal{R}, d)=\mathcal{K} \otimes H(\g, d)$
admits a nonlinear supersymmetric LCA structure, then
\[ H(V(\mathcal{R}), d)\simeq V (H(\mathcal{R}, d)).\]
\end{enumerate}
\end{theorem}

\section{BRST cohomology}
We are now in a position to define supersymmetric $W$-algebras
via BRST cohomology following \cite{MR94}. We will rely on the
supersymmetric vertex algebra theory developed by Heluani and Kac~\cite{HK06,K96}
to describe the structure of the $W$-algebras
associated with odd nilpotent elements of Lie superalgebras.

\subsection{BRST complex}
Let $\g$ be a finite-dimensional simple Lie superalgebra with a $(\frac12\Z)$-grading
$\g= \bigoplus_{i\in \frac12\Z} \g(i)$ satisfying the following conditions:
\begin{enumerate}[(i)]
\item There exists $h\in \g_{\bar{0}}$ such that $\g(i)= \{a\in \g\ts |\ts \frac{1}{2} [h,a]= ia\}$.
\item There are odd elements $f_{\text{\rm odd}}\in \g(-\frac{1}{2})$ and
$e_{\text{\rm odd}}\in \g(\frac{1}{2})$  such that
\[ \text{span}\{e,e_{\text{\rm odd}}, h, f_{\text{\rm odd}}, f\}\simeq {\mathfrak{osp}}(1|2),\]
where $(e,h,f)$ is an $\sll_2$-triple.
\end{enumerate}
We will suppose that $\g$ is equipped with
a nondegenerate invariant bilinear form $(\, |\, )$
normalized by the conditions $(e|f)= \frac{1}{2} (h|h)=1$.

Introduce two supersymmetric vertex algebras.
\begin{enumerate}
\item Let $\overline{\g}=\{ \overline{a}\ts|\ts a\in \g\}$ be the vector superspace
defined by $\overline{\g}_{\bar{1}} = \g_{\bar{0}}$ and
$\overline{\g}_{\bar{0}} = \g_{\bar{1}}.$ The
{\it supersymmetric current nonlinear LCA}  is
\[ \mathcal{R}_{\text{cur}}:= \mathcal{K} \otimes \overline{\g}\]
endowed with the $\Lambda$-bracket
\[ [ \overline{a}_\Lambda \overline{b}] =(-1)^{p(a)p(\overline{b})}
\overline{[a,b]}+ k\ts\chi(a|b).\]
\item  Set $\n = \bigoplus_{i>0} \g(i)$ and $\n_- = \bigoplus_{i<0} \g(i)$.
Then there are bases \[ \{u_\alpha\ts|\ts \alpha\in I_+\} \quad \text{ and }
\quad  \{u^\alpha\ts|\ts \alpha\in I_+\}\] of $\n$ and $\n_-$, respectively,
parameterized by a certain index set $I_+$, such that
$(u^\alpha|u_\beta)= \delta_{\alpha,\beta}$. Introduce two vector superspaces
\[ \phi_\n \simeq \n \subset \g, \qquad \phi^{\overline{\n}_-}
\simeq \overline{\n}_- \subset \overline \g,
\]
spanned by the respective families of elements $\phi_b$ and $\phi^{\overline{a}}$
with $b\in \n$ and $\overline{a}\in \overline{\n}_-$.
Consider the supersymmetric nonlinear LCA $\mathcal{R}_{\text{ch}}
=\mathcal{K} \otimes (\phi_\n \oplus  \phi^{\overline{\n}_-})$
endowed with the $\Lambda$-bracket
\[ [ \phi^{\overline{a}}\, _\Lambda \phi_b]= [\phi_{b\, \Lambda}
\phi^{\overline{a}}]= (a|b).\]

\end{enumerate}
Due to the results of Section~\ref{subsec:snlca}, the two above
supersymmetric nonlinear LCAs give rise to
respective universal enveloping supersymmetric vertex algebras
$V(\mathcal{R}_{\text{cur}})$ and $V(\mathcal{R}_{\text{ch}})$.
Their tensor product
\ben
C(\overline{\g}, f_{\text{\rm odd}}, k) =  V(\mathcal{R}_{\text{cur}})\otimes
V(\mathcal{R}_{\text{ch}})
\een
also carries
a supersymmetric vertex algebra structure. Introduce the element $d$ by
\begin{equation}\label{Eqn:BRST_d}
d= \sum_{\alpha\in I_+} : (\overline{u}_\alpha-(f_{\text{\rm odd}}|u_\alpha))
\phi^\alpha:+\frac{1}{2} \sum_{\alpha,\beta\in I_+} (-1)^{p(\alpha)
p(\overline{\beta})}: \phi_{[u_\alpha, u_\beta]} \phi^\beta \phi^\alpha:\, ,
\end{equation}
where $\phi^\alpha= \phi^{\overline{u}^\alpha}$,
$\phi_\alpha= \phi_{u_\alpha}$, $p(\alpha)=p(u_\alpha)$
and $p(\overline{\alpha}) = p(\overline{u}_\alpha)$.

\begin{proposition} \label{Prop:d-Lambda-C}
The $\Lambda$-brackets between $d$ and elements in
$C(\overline{\g}, f_{\text{\rm odd}}, k)$ have the form:
\ben
\bal[]
[d_\Lambda \overline{a}]&= \sum_{\alpha\in I_+}
(-1)^{p(\overline{a})p(\alpha)} :\phi^\alpha \overline{[u_\alpha, a]}:
+ \sum_{\alpha\in I_+}(-1)^{p(\overline{\alpha})} k (\chi+S) \phi^\alpha (u_\alpha|a),\\
[d_\Lambda \phi^\alpha]&= \frac{1}{2}
\sum_{\alpha,\beta\in I_+} (-1)^{p(\overline{\alpha})p(\beta)}
:\phi^\beta \phi^{\overline{[u_\beta, u^\alpha]}}:,\\
[d_\Lambda \phi_\alpha] & = (-1)^{p(\overline{\alpha})}
u_\alpha -(f_{\text{\rm odd}}|u_\alpha) + \sum_{\beta\in I_+}
(-1)^{p(\overline{\alpha}) p(\beta)} :\phi^\beta \phi_{[u_\beta, u_\alpha]}:.
\eal
\een
\end{proposition}

\begin{proof}
The formulas are verified by a direct calculation in the same way as
for the supersymmetric classical $W$-algebras; see \cite{KS18}.
\end{proof}

Set $Q:=d_{(0|0)}$. Then, by the Wick formula \eqref{Wick}, we have
 \begin{equation} \label{d_property}
Q ( :A\, B:) = : Q(A) \, B:+ (-1)^{p(A)} :A \,Q(B): \, .
 \end{equation}

\begin{proposition}
The linear map $ Q$ on $C(\overline{\g}, f_{\text{\rm odd}}, k)$ satisfies $ Q^2=0$.
\end{proposition}

\begin{proof}
This follows by a direct computation with the use of
Proposition~\ref{Prop:d-Lambda-C} and property \eqref{d_property}.
\end{proof}

By taking the cohomology of the BRST
complex $C(\overline{\g}, f_{\text{\rm odd}}, k)$ with the differential $Q$,
we can now define the corresponding
supersymmetric $W$-algebra as in \cite{MR94}; cf. \cite{a:iw} and
\cite[Ch.~15]{fb:va}.

\begin{definition}
The {\em supersymmetric $W$-algebra associated to $\overline{\g}$, $f_{\text{\rm odd}}$ and $k\in \C$} is
\[ W(\overline{\g}, f_{\text{\rm odd}}, k)= H(C(\overline{\g}, f_{\text{\rm odd}}, k),Q).\]
\end{definition}

\begin{proposition} Let $A,B\in C(\overline{\g}, f_{\text{\rm odd}}, k)$
satisfy $ Q(A)= Q(B)=0$ and $C$ be any element in $C(\overline{\g}, f_{\text{\rm odd}}, k)$.
Then the following holds:
\begin{enumerate}
\item  $ Q(SA)= Q(:AB:)=0$ and $ Q([A_\Lambda B])=0$;
\item  $S(QC)$, $: Q(C)\, B:$ and $[ Q(C)_\Lambda B]$ belong to the image of $ Q.$
\end{enumerate}
\end{proposition}
\begin{proof}
By sesquilinearity of supersymmetric LCAs, for any
$X\in  C(\overline{\g}, f_{\text{\rm odd}}, k)$ we have $S( QX)=- Q(SX)$.
Hence the first properties in (1) and (2) hold. The second properties follow from \eqref{d_property}.
By the Jacobi identity of
supersymmetric LCAs, for $X,Y\in  C(\overline{\g}, f_{\text{\rm odd}}, k)$ we have
\ben
Q([X_\Lambda Y]) = -[Q(X)_\Lambda Y] +(-1)^{p(X)+1}[X_\Lambda Q(Y)]
\een
which gives the third properties in (1) and (2).
\end{proof}

\begin{corollary}
The supersymmetric $W$-algebra $W(\overline{\g}, f_{\text{\rm odd}}, k)$
is a supersymmetric vertex algebra.
\end{corollary}

\subsection{Building blocks of supersymmetric $W$-algebras}
For any $\bar{a}\in \bar{\g}$ set
\ben
J_{\bar{a}}= \bar{a}+\sum_{\beta\in I_+} (-1)^{p(\bar{a})p(\bar{\beta})}
:\phi^\beta \phi_{[u_\beta,a]}: \, \in C(\overline{\g}, f_{\text{\rm odd}}, k).
\een

\begin{proposition} \label{Prop: d(J)}
For the element $d$ defined in \eqref{Eqn:BRST_d} we have
\[ [ d_\Lambda J_{\bar{a}}]= \sum_{\beta\in I_+} (-1)^{p(\overline{a})p(\beta)}
:\phi^\beta (J_{\overline{\pi_{\leqslant 0}[u_\beta,a]}}+(f_{\text{\rm odd}}|[u_\beta, a])):
+\sum_{\beta \in I_+}(-1)^{\overline{\beta}}k\tss (S+\chi)\tss \phi^\beta(u_\beta|a), \]
where $\pi_{\leqslant 0}: \g \to \oplus_{i\leqslant 0} \g(i)$ is the projection map
with the kernel $\oplus_{i> 0} \g(i)$.
\end{proposition}
\begin{proof} By the Wick formula,
\begin{align}
\, [d_\Lambda J_{\bar{a}}]& =   [d_\Lambda \bar{a}]+ \sum_{\beta\in I_+}
(-1)^{p(\bar{a})p(\bar{\beta})}[d \, _\Lambda :\phi^\beta \phi_{[u_\beta,a]}:] \nonumber \\
&{} =  [d_\Lambda \bar{a}] +  \sum_{\beta\in I_+}
 (-1)^{p(\bar{a})p(\bar{\beta})}: [d \, _\Lambda \phi^\beta] \phi_{[u_\beta,a]}:
  \label{(3.5)_1105}\\
 & + \sum_{\beta,\gamma\in I_+,\, k\geqslant 1} \frac{\lambda^k}{2k!}
 (-1)^{p(\bar{\beta})(p(\gamma)+p(a)+1)}  \left(:\phi^\gamma
 \phi^{\overline{[u_\gamma, u^\beta]}}:\right)_{(k-1|1)} \phi_{[u_\beta,a]}\label{(3.6)_1105}\\[0.5em]
 & + \sum_{\beta\in I_+}(-1)^{p(\bar{a})p(\bar{\beta})}
 :\phi^\beta[d_\Lambda \phi_{[u_\beta, a]}]:.
 \non
\end{align}
 Since the coefficients of $\Lambda^{j_0} \chi$ in
 $[\phi_{[u_\beta, a]\ \Lambda} : \phi^\gamma \phi^{\overline{[u_\gamma, u^\beta]}}:]$
 are all zero, the coefficients of $\Lambda^{j_0} \chi$ in
\[\, [: \phi^\gamma \phi^{\overline{[u_\gamma, u^\beta]}}:
\,  _\Lambda \phi_{[u_\beta, a]} ] = (-1)^{p(\beta)p(\bar{a})}
[\phi_{[u_\beta, a]\ -\Lambda-\nabla} : \phi^\gamma
\phi^{\overline{[u_\gamma, u^\beta]}}:] \, \]
are also $0$ so that the expression in $\eqref{(3.6)_1105}$ vanishes.
The second term in \eqref{(3.5)_1105} equals
\[  \sum_{\beta,\gamma\in I_+} \frac{1}{2}
(-1)^{p(\bar{\beta})(p(\gamma)+p(a)+1)} :\ts :\phi^\gamma
\phi^{\overline{[u_\gamma, u^\beta]}}:\ts\phi_{[u_\beta,a]}:.\]
By the quasi-associativity in \eqref{Eqn:q-comm/assoc} and the fact
that $\phi^{\bar{n}}\, _{(j|1)}\phi_{m}=0$ for any $n\in \n$ and $m\in \n_-$
with $j\geqslant 0$, we have
\[ :\ts :\phi^\gamma \phi^{\overline{[u_\gamma, u^\beta]}} :\phi_{[u_\beta,a]}:\ts =\ts
  : \phi^\gamma  : \phi^{\overline{[u_\gamma, u^\beta]}} \phi_{[u_\beta,a]}:\ts :\ts .\]
The remaining computations are straightforward, they are
analogous to the classical case in \cite{KS18}.
\end{proof}

\begin{proposition}
If $a,b\in \bigoplus_{i\leqslant 0} \g(i)$ or $a,b\in \bigoplus_{i>0}\g$ then
\[\,[ J_{\bar{a}\, \Lambda } J_{\bar{b}}]=(-1)^{p(a)p(\bar{b})}J_{\overline{[a,b]}}+k\tss(S+\chi)\tss(a|b).\]
\end{proposition}
\begin{proof}
This is verified by a direct computation.
\end{proof}

Introduce the vector superspaces
\[ r_+= \phi_\n \oplus J_{\bar{\n}}\Fand r_-= J_{\bar{\g}_{\leqslant 0}}\oplus \phi^{\bar{\n}_-},\]
where
\ben
J_{\bar{\n}}=\text{span}\ts
\{ J_b\ts |\ts b\in \bar{\n}\}\Fand
J_{\bar{\g}_{\leqslant 0}} = \text{span}\ts
\{ J_{\bar{a}}\ts |\ts a\in \bigoplus_{i\in \Z_{\leqslant 0}}\g(i)\}.
\een
It is not difficult to see that
both $\mathcal{R}_+= \mathcal{K} \otimes r_+$ and $\mathcal{R}_-= \mathcal{K} \otimes r_-$
are supersymmetric nonlinear LCAs and that $C(\bar{\g}, f_{\text{\rm odd}}, k)$ decomposes
into the tensor product of supersymmetric
vertex subalgebras:
\[ C(\bar{\g}, f_{\text{\rm odd}}, k)=V(\mathcal{R}_+) \otimes V(\mathcal{R}_-).\]

\begin{lemma}[K\"unneth lemma]
Let $V_1$ and $V_2$ be vector superspaces and $d_i:V_i\to V_i$, $i=1,2$,
be differentials. If $d:V_1\otimes V_2 \to V_1 \otimes V_2$ is
defined by
\ben
d(a\otimes b)= d_1(a)\otimes b +(-1)^{p(a)} a\otimes d_2(b)
\een
then
\[ H(V,d) \simeq H(V_1, d_1) \otimes H(V_2, d_2).\]
\end{lemma}

\begin{proposition}
The differential $ Q$ has the properties
\begin{equation}\label{differential_decomp}
 Q(V(\mathcal{R}_+)) \subset V(\mathcal{R}_+)
 \fand Q(V(\mathcal{R}_-)) \subset V(\mathcal{R}_-),
\end{equation}
so that
\begin{equation}\label{apply Kunneth to W-algebra}
W(\bar{\g}, f_{\text{\rm odd}}, k)= H(V(\mathcal{R}_+), Q) \otimes H(V(\mathcal{R}_-), Q).
\end{equation}
\end{proposition}

\begin{proof}
The inclusions
\eqref{differential_decomp} follow from Propositions~\ref{Prop:d-Lambda-C}
and~\ref{Prop: d(J)}. The decomposition \eqref{apply Kunneth to W-algebra} is then implied by
the K\"unneth lemma.
\end{proof}

\subsection{Generators of supersymmetric $W$-algebras}

We now aim to describe the cohomologies
$H(V(\mathcal{R}_+), Q)$
and $H(V(\mathcal{R}_-), Q)$.

\begin{proposition}
We have
$H (V(\mathcal{R}_+), Q)= \C$
so that $W(\bar{\g}, f_{\text{\rm odd}}, k)= H (V(\mathcal{R}_-), Q). $
\end{proposition}
\begin{proof}
Set $K_{\bar{n}}= (-1)^{p(\bar{n})}J_{\bar{n}}-(f_{\text{\rm odd}}|n)$ for $n\in \n$
and introduce the superspace
\[r'_+= \phi_\n \oplus K_{\bar{\n}},\qquad
K_{\bar{\n}}=\text{span}\ts
\{ K_{\bar{n}}\ts |\ts \bar{n}\in \bar{\n}\}.
\]
Then $\mathcal{R}_+= \mathcal{K}\otimes r'_+$. Define the conformal
weight $\Delta$ and the bigrading on $r'_+$ by
\[ \Delta(\phi_n)= \Delta(K_{\bar{n}})= j_n, \quad \text{gr}(\phi_n)
= (j_n-1, -j_n), \quad \text{gr}(K_{\bar{n}})= (j_n-1, -j_n+1),\]
 assuming that $n\in  \g(j_n)$.
The graded differential $Q^{\tss\text{gr}}$ associated with $ Q$ is good
almost linear (see Section \ref{Sec:filtered complex}) and
\[ H({r'}_+, Q^{\tss\text{gr}} )=0.\]
By Theorem \ref{Thm:cohomology}, we have $H (V(\mathcal{R}_+), Q)= \C$.
\end{proof}

\medskip

To describe $H(V(\mathcal{R}_-), Q)$, recall that
\begin{align}
Q(J_{\bar{a}})   =   \sum_{\beta\in I_+}(-1)^{p(\bar{a})p(\beta)} &{}:
 \phi^\beta(J_{\overline{\pi_{\leqslant 0} [u_\beta, a] }} +(f_{\text{\rm odd}}|[u_\beta, a])) :
 \non\\
{}&+ \sum_{\beta\in I_+} (-1)^{p(\bar{\beta})} k S\phi^\beta(u_\beta|a)
\label{d(r-)}
\end{align}
and 
\beql{d(r-)two}
Q(\phi^{\bar{m}})= \frac{1}{2} \sum_{\beta\in I_+}
(-1)^{p(\bar{m})p(\beta)} :\phi^\beta \phi^{\overline{[u_\beta, m]}}:.
\eeq

Consider the conformal weight $\Delta$ and the bigrading on $r_-$
satisfying
\ben
\begin{aligned}
& \Delta(J_{\bar{a}})= \frac{1}{2} -j_a, \quad \Delta(\phi^{\bar{m}})= -j_m,\\
& \text{gr}(J_{\bar{a}})= (j_a, -j_a), \quad \text{gr}(\phi^{\bar{m}})
= \big(j_m+\frac{1}{2}, -j_m+\frac{1}{2}\big),
\end{aligned}
\een
where $a\in  \g(j_a)$ and $m\in \g(j_m)$ for $j_a\leqslant 0$ and $j_m<0$.
Note that
\[ \Delta(\phi^\beta)= j_\beta, \quad \text{gr}(\phi^\beta)=
\big(-j_\beta+\frac{1}{2}, j_\beta+\frac{1}{2}\big),\]
where $u^\beta \in \g(-j_\beta).$ Since $\Delta(S)= \frac{1}{2}$
and $\text{gr}(S)= (0,0)$. Every term in \eqref{d(r-)}
has conformal weight $\frac{1}{2}-j_a$ and every term in \eqref{d(r-)two}
has conformal weight $-j_m$. The
bigradings of terms in \eqref{d(r-)} are given by
\begin{equation}\label{bigrading_first}
\begin{aligned}
&  \text{gr}(:\phi^\beta J_{\overline{\pi_{\leqslant 0}[u_\beta, a]}}:)
=\big(j_a+\frac{1}{2}, -j_a+\frac{1}{2}\big) , \\
& \text{gr}( \phi^\beta(f_{\text{\rm odd}}|[u_\beta,a]))= (j_a, -j_a+1),  \\
& \text{gr}(S\phi^\beta(u_\beta|a))=\big (j_a+\frac{1}{2}, -j_a+\frac{1}{2} \big).
\end{aligned}
\end{equation}
The bigradings of terms in \eqref{d(r-)two} are
\begin{equation}\label{bigrading_second}
\text{gr}(\phi^{\bar{m}})=\big (j_m+\frac{1}{2}, -j_m+\frac{1}{2}\big),
\qquad \text{gr}(:\phi^\beta \phi^{\overline{[u_\beta, m]}}:)= (j_m+1, -j_m+1).
\end{equation}

\begin{theorem} \label{Thm_Sec_4}
Let $\Ker\, (\ad \, f_{\text{\rm odd}})=
\{\, u_\alpha\, | \, \alpha\in \mathcal{J}\, \}$ with an index set $\mathcal{J}$.
Then
\begin{enumerate}
\item $W(\bar{\g}, f_{\text{\rm odd}}, k)$ is freely generated by $|\mathcal{J}|$
elements as a differential algebra,
\item there exists a free generating set of the form
\[ \{\, u_\alpha+ A_\alpha \, |\, \alpha\in \mathcal{J} \, \},\]
where  $A_\alpha \in F^{j_\alpha+\frac{1}{2}}
V(\mathcal{\mathcal{R}_-})^0[\frac{1}{2}-j_\alpha]$ for $u_\alpha\in \g(j_\alpha)$.
\end{enumerate}
\end{theorem}

\begin{proof}
Since we know that $W(\overline{\g},f_{\text{\rm odd}}, k)= H(V(\mathcal{R}_-), Q)$,
it is enough to show (1) and (2) for $H(V(\mathcal{R}_-), Q).$
The conformal weight and bigrading on $r_-$ induce those
on $V(\mathcal{R}_-)$. With respect to the conformal weight
and bigrading, $Q$ induces the graded differential $ Q^{\tss\text{gr}}$.
The bigradings listed in \eqref{bigrading_first} and \eqref{bigrading_second}
show that
\[ Q^{\tss\text{gr}}(J_{\bar{a}})= \sum_{\beta\in I_+}
(-1)^{p(\bar{a})p(\beta)}\phi^\beta (f_{\text{\rm odd}}|[u_\beta, a]),
\quad Q^\text{gr}(\phi^{\bar{m}})=0.\]
Note that  $V(\mathcal{R}_-)^0 \cap r_-=J_{\g_{\leqslant 0}}$ and
$V(\mathcal{R}_-)^1 \cap r_-=\phi^{\bar{\n}_-}.$
Since $ Q^{\tss\text{gr}}(r_-)=\phi^{\bar{\n}_-}$, we have
$H^{p,q}(r_-, Q^{\tss\text{gr}})=0$ when $p+q\neq 0$ and so $ Q$
is a good almost linear differential map.
Furthermore,
$ \Ker( Q^{\tss\text{gr}}|_{r_-})=\{ J_a | a\in \Ker(\ad \, f_{\text{\rm odd}})\}
\oplus \phi^{\bar{\n}_-}$,
hence \[H(r_-, Q^{\tss\text{gr}})=\{ J_a\ts |\ts a\in \Ker(\ad \, f_{\text{\rm odd}})\}.\]
Thus, using Theorem \ref{Thm:cohomology}, we arrive at (1) and (2).
\end{proof}

\section{Generators of $W(\overline{\g}, f_{\text{\rm prin}}, k)$ for $\g=\gl(n+1|n)$}

Consider the Lie superalgebra $\g= \gl(n+1|n)$
with the basis $\{E_{i,j}|i,j=1, \dots, 2n+1\}$ and the $\Z/2\Z$-grading defined by
$p(E_{i,j})= i+j\mod 2$ with the commutation relations
\[ [E_{i,j}, E_{i',j'}]= \delta_{j,i'}E_{i,j'}-(-1)^{(i+j)(i'+j')}\ts\delta_{i,j'}\ts E_{i',j}.\]
Take the odd principal nilpotent element in the form
\[ f_{\text{\rm prin}}= \sum_{p=1}^{2n} E_{p+1,p}.\]
By Proposition \ref{Prop: d(J)}, for $C(\bar{\g}, f_{\text{\rm prin}}, k)$ and
 any $m\geqslant l$, we have
\ben
\begin{aligned}
 Q(J_{m,l}) & = (-1)^m k\tss S\tss \phi^{l,m}+ \sum_{j=l+1}^{m} (-1)^{l+j+1} :\phi^{l,j}J_{m,j}: \\
 &+ \sum_{i=l}^{m-1} (-1)^{(i+m)(m+l+1)}:\phi^{i,m}J_{i,l}:
 + (-1)^l \phi^{l, m+1} +(-1)^m\phi^{l-1,m},
 \end{aligned}
 \een
 where we set $\phi^{j,i}= (-1)^{i+1} \phi^{\overline{E_{ij}}}$
 for $i>j$ and
 $J_{i,j}= J_{\overline{E_{i,j}}}$  for $i\geqslant j$.

 We will be working with operators on $C(\bar{\g}, f_{\text{\rm prin}}, k)$
 of the form $\sum_{t=0}^N A_t S^t$ with $A_t \in C(\bar{\g}, f_{\text{\rm prin}}, k)$,
 which act on an arbitrary element $X\in C(\bar{\g}, f_{\text{\rm prin}}, k)$ by the rule
 \[ \sum_{t=0}^N A_t S^t(X)= \sum_{t=0}^N :A_t (S^t (X)):  .\]
In particular, for the operator $A_{i,j}= \delta_{ij}\tss k\tss S +(-1)^{i+1} J_{i,j}$
on $C(\bar{\g}, f_{\text{\rm prin}}, k)$ we have
 \ben
 A_{i,j} (X)= \delta_{ij}\tss k\tss S(X) +(-1)^{i+1} : J_{i,j} X:\, .
 \een

 Consider the $(2n+1)\times (2n+1)$ matrix
 \ben
 \mathcal{A}:= \begin{bmatrix}
 A_{1,1} & -1 & 0 &  \cdots  &   \cdots & 0  \\
 A_{2,1} & A_{2,2} &  -1&\cdots  &  \cdots & 0  \\
 \vdots &  \vdots & \vdots   & \vdots  & \vdots   & \vdots \\
 A_{2n, 1} & A_{2n,2} & A_{2n,3} & \cdots    & A_{2n, 2n} & -1\\
  A_{2n+1, 1} & A_{2n+1,2} & A_{2n+1,3} & \cdots   & A_{2n+1, 2n} & A_{2n+1, 2n+1}
 \end{bmatrix}
 \een
whose entries are operators on $C(\bar{\g}, f_{\text{\rm prin}}, k)$.
Then the {\em column} (or {\em row}) {\em determinant} of $\Ac$ is
given by the formula
\begin{equation} \label{column deteminant}
\cdet\ts\mathcal{A} = \sum_{N=0}^{2n}\quad\sum_{0=i_0<i_1<\dots<
 i_{N+1}= 2n+1} A_{i_1, i_0+1} A_{i_2, i_1+1}
 \dots A_{i_{N+1}, i_N+1}.
 \end{equation}
 Write
 \[ \cdet\ts\mathcal{A}= W_0+W_1 S+\dots +W_{2n+1} S^{2n+1}\]
 for certain elements $W_p \in C(\bar{\g}, f_{\text{\rm prin}}, k)$. Clearly, $W_{2n+1}=k^{2n+1}$.

 \begin{theorem}\label{theorem1}
 All elements $W_1,\dots,W_{2n}$ belong to the $W$-algebra $W(\bar{\g},f_{\text{\rm prin}}, k)$.
 \end{theorem}

 \begin{proof}
One readily verifies that
 \[ Q  \sum_{p=0}^{2n+1}  W_p S^p = \sum_{p=0}^{2n+1} Q(W_p) S^p-W_p S^p Q\]
 so that $ Q A_{m,l} = (-1)^{m+l+1} A_{m,l} Q + (-1)^{m+1} Q(J_{m,l}) $. Therefore,
 \begin{multline} \non
  Q \, A_{i_1, i_0+1} \dots A_{i_{p+1}, i_p+1} \dots A_{i_{N+1}, i_N+1} \\
 =   \sum_{p=0}^N  (-1)^{i_p}\big(A_{i_1, i_0+1} \dots \big((-1)^{i_{p+1}+1}
 Q(J_{i_{p+1},i_p+1 })\big)\dots A_{i_{N+1}, i_N+1}\big)\\
 - A_{i_1, i_0+1} \dots A_{i_{p+1}, i_p+1} \dots A_{i_{N+1}, i_N+1}Q.
 \end{multline}
Hence the property $W_p\in W(\g,f_{\text{\rm prin}}, k)$ will follow if we
show that $ \sum_{N=0}^{2n} B_N=0$, where we set
\ben
B_N= \sum_{p=0}^N  (-1)^{i_p}\big(A_{i_1, i_0+1} \dots \big((-1)^{i_{p+1}+1}
Q(J_{i_{p+1},i_p+1 })\big)\dots A_{i_{N+1}i_N+1}\big).
  \een
Using the relations
\[ J_{i,j}= (-1)^{i+1} (A_{i,j}-\delta_{i,j} kS)
\fand :\phi^{j,i}J_{i',j'}:= (-1)^{(i+j+1)(i'+j'+1)} :J_{i',j'}\phi^{j,i}:\]
we find that
\begin{multline} \non
(-1)^{i_{p+1}+1} Q(J_{i_{p+1}, i_p+1})\\
{}= -kS(\phi^{ i_p+1,i_{p+1}})
+ \sum_{j=i_p+2}^{i_{p+1}} (-1)^{i_p+j} \phi^{i_p+1,j} \big(A_{i_{p+1}, j}-\delta_{i_{p+1}, j}kS\big)\\
 + \sum_{i=i_p+1}^{i_{p+1}-1}(-1)^{i_p+i}\big(A_{i, i_p+1}-\delta_{i, i_p+1}k S\big)
\phi^{i,i_{p+1}}+(-1)^{i_p+i_{p+1}} \phi^{i_p+1, i_{p+1}+1}-\phi^{i_p,i_{p+1}}
\end{multline}
and
\ben
 -kS(\phi^{ i_p+1, i_{p+1}})+(-1)^{i_p+i_{p+1}+1}
 \phi^{ i_p+1, i_{p+1}}S+S\phi^{i_p+1,i_{p+1}}=0.
\een
Therefore,
\ben
\bal
(-1)^{i_{p+1}+1} Q(J_{i_{p+1}, i_p+1})& =  \sum_{j=i_p+2}^{i_{p+1}} (-1)^{i_p+j}
\phi^{i_p+1, j} A_{i_{p+1}, j} \\
&{} + \sum_{i=i_p+1}^{i_{p+1}-1}(-1)^{i_p+i}A_{i, i_p+1} \phi^{i,i_{p+1}}+(-1)^{i_p+i_{p+1}}
\phi^{ i_p+1, i_{p+1}+1}-\phi^{i_p,i_{p+1}}
\eal
\een
so that $B_N$ can be expressed as
\ben
\begin{aligned}
& \sum_{p=0}^{N}  A_{i_1, i_0+1} \dots A_{i_p, i_{p-1}+1}
\Bigg[ \Big( \sum_{j=i_p+2}^{i_{p+1}} (-1)^{j} \phi^{ i_p+1,j}
A_{i_{p+1}, j} +(-1)^{i_{p+1}} \phi^{ i_p+1, i_{p+1}+1} \Big) \\
& +\Big(\sum_{i=i_p+1}^{i_{p+1}-1}(-1)^{i}A_{i, i_p+1}
\phi^{i,i_{p+1}}-(-1)^{i_p}\phi^{i_p,i_{p+1}} \Big) \Bigg]
A_{i_{p+2}, i_{p+1}+1}\dots A_{i_{N+1}, i_N+1}.
\end{aligned}
\een
By the quasi-associativity property, we have
\begin{equation*}
\begin{aligned}
 ( \phi^{i_p+1,j} A_{i_{p+1}, j})(A_{i_{p+2}, i_{p+1}+1}\dots A_{i_{N+1}, i_N+1} )&=
\phi^{ i_p+1,j} (A_{i_{p+1}, j}(A_{i_{p+2}, i_{p+1}+1}\dots A_{i_{N+1}, i_N+1})),\\[0.4em]
 (A_{i, i_p+1} \phi^{ i,i_{p+1}})(A_{i_{p+2}, i_{p+1}+1}\dots A_{i_{N+1}, i_N+1})
&= A_{i, i_p+1}( \phi^{ i,i_{p+1}}(A_{i_{p+2}, i_{p+1}+1}\dots A_{i_{N+1}, i_N+1}))
\end{aligned}
\end{equation*}
for $j= i_p+2, \dots, i_{p+1}$ and $i=i_p+1, \dots, i_{p+1}$, so that vanishing of the telescoping sum
implies that $\sum_{N=0}^{2n} B_N=0$.
 \end{proof}

 \begin{lemma} \label{Lem:Sec5}
Suppose that $\{v_p\ts |\ts p=0, \dots, 2n\}$ is a basis of
 $\Ker\,(\ad f_{\text{\rm odd}})$ such that
 $\Delta_{J_{\bar{v}_p}}= \frac{1}{2} (2n+1-p)$.
 Take $V_p\in W(\overline{\g}, f_{\text{\rm prin}}, k)$ of the form
$V_p= J_{\bar{v}_p}+ w_p$
satisfying the conditions
 \begin{enumerate}[(i)]
 \item $V_p$ and $w_p$ have the conformal weight $ \frac{1}{2} (2n+1-p)$,
 \item $w_p$ lies in the differential algebra generated by $J_{\bar{a}}$
 for $\Delta_{J_{\bar{a}}}<\Delta_{V_p}$.
 \end{enumerate}
 Then the set $\{V_p\ts |\ts p=0, \dots, 2n\}$ freely
 generates the $W$-algebra $W(\overline{\g}, f_{\text{\rm prin}}, k)$.
 \end{lemma}
 \begin{proof}
A generating set of the form
 $\{V'_p=J_{\bar{v}_p}+ w'_p\ts |\ts p=0, \dots, 2n\}$ satisfying the required conditions
(i) and (ii) exists by Theorem \ref{Thm_Sec_4}.
Set
 \[ \Wc_m:= \text{subalgebra freely generated by $\{V_m,V_{m+1}, \dots, V_{2n}\}$},\]
  \[ \Wc^{\tss\prime}_m:= \text{subalgebra freely generated by $\{V'_m,V'_{m+1}, \dots, V'_{2n}\}$}.\]
  We will show by a (reverse) induction that $\Wc_m=\Wc^{\tss\prime}_m$ for all $m=0,\dots,2n$.
Note that $\Wc_{2n}= \Wc^{\tss\prime}_{2n}$, since $w_{2n}$ and $w'_{2n}$ are constants.
Now suppose that $\Wc_{p}= \Wc^{\tss\prime}_{p}$ for some $p\leqslant 2n$.
Then $V_{p-1}-V'_{p-1}\in \Wc_p= \Wc^{\tss\prime}_p$
 by condition (ii).
 Hence we can conclude that
 $V'_{p-1}= V_{p-1}+ (w'_p-w_p) \in  \Wc_{p-1}$ and, similarly, $V_{p-1}\in \Wc^{\tss\prime}_{p-1}$.
 This shows that $\Wc_{p-1}= \Wc^{\tss\prime}_{p-1}$. Thus, $ \Wc^{\tss\prime}_0= \Wc_0$ and since
$W(\overline{\g}, f_{\text{\rm prin}}, k)= \Wc^{\tss\prime}_0$, the lemma follows.
 \end{proof}

 \begin{theorem}
 \label{thm:freegen}
The set of coefficients
 $\{ W_p\ts|\ts p=0, \dots,2n\} $ of $\cdet\ts\Ac$ freely generates $W(\overline{\g},f_{\text{\rm prin}}, k)$
 as a differential algebra.
 \end{theorem}
 \begin{proof}
 Note that for $i\geqslant j$ we have
 \[ \Delta_{A_{i,j}(X)}= \frac{1}{2}(i-j+1)+\Delta_X,\]
and each term in \eqref{column deteminant} satisfies
 \[ \Delta_{ A_{i_1, i_0+1} A_{i_2, i_1+1} \dots A_{i_{N+1}, i_N+1}(X)}
 = \frac{2n+1}{2} +\Delta_X.\]
 A direct calculation gives
 \ben
 W_{2n-k} = \sum_{l=1}^{2n+1-k} (-1)^{kl} J_{k+l, l} + w_{2n-k}
\qquad\text{for}\quad {k=0,1, \dots, 2n},
\een
  where $\Delta_{2n-k}=  \frac{2n+1}{2}-\frac{2n-k}{2}$ and $w_{2n-k}$
  can be expressed as a normally ordered product of
  the elements $J_{i,j}$ with $0\leqslant i-j\leqslant k$
  and their derivatives. It remains to apply Lemma \ref{Lem:Sec5}.
 \end{proof}

 \begin{example}
 Let $\g= \gl(2|1)$. Then $f_{\text{\rm prin}}= E_{21}+E_{32}$ and
 \[ \mathcal{A}= \left[ \begin{array}{ccc} A_{1,1} & -1 &  0 \\
 A_{2,1} & A_{2,2} & -1 \\ A_{3,1} & A_{3,2} & A_{3,3}  \end{array}\right].\]
 The column determinant of $\mathcal{A}$ is
 \ben
 \begin{aligned}
\cdet\ts\Ac& = A_{1,1} A_{2,2}A_{3,3}+ A_{3,1}+A_{2,1}A_{3,3} + A_{1,1}A_{3,2} \\
 & = (kS)^3 + W_2 S^2 +W_1 S+W_0.
 \end{aligned}
 \een
 where
 \ben
 \begin{aligned}
  W_2& = k^2(J_{1,1}+J_{2,2}+J_{3,3}), \\
  W_1& = k(-J_{1,1} J_{2,2} -J_{1,1}J_{3,3}-J_{2,2}J_{3,3} -J_{2,1} +J_{3,2} -kJ'_{2,2}), \\
  W_0& = -J_{1,1}J_{2,2}J_{3,3} -J_{2,1} J_{3,3} +J_{1,1} J_{3,2}+J_{3,1} \\
 & +kJ'_{3,2}+kJ_{1,1}J'_{3,3}-kJ'_{2,2}J_{3,3} +kJ_{2,2}J'_{3,3} + k^2 J''_{3,3},
 \end{aligned}
 \een
and $X':=[S,X]$.
 Hence $W(\overline{\g}, f_{\text{\rm prin}}, k)$ is freely generated by $W_0, W_1$ and $W_2.$
 \qed
 \end{example}

As in \cite{am:eg},
by taking the quotient of the $W$-algebra $W(\overline{\g}, f_{\text{\rm prin}}, k)$
over the supersymmetric
vertex algebra ideal generated by the elements $J_{i,j}$ with $i>j$ we recover the
presentation of the $W$-algebra via the
{\em Miura transformation}; cf. \cite{eh:st,i:qh,i:ns}:
\ben
\cdet\ts\Ac\mapsto \big(kS + J_{1,1}\big)\big(kS - J_{2,2}\big)\big(kS + J_{3,3}\big)
\dots \big(kS - J_{2n,2n}\big)\big(kS + J_{2n+1,2n+1}\big).
\een

\end{document}